\documentclass[smallcondensed]{svjour3} 
\pdfoutput=1

\def\PubType{Proof} 

\usepackage[utf8]{inputenc}
\usepackage[english]{babel}  
\usepackage[T1]{fontenc} 
\usepackage{IEEEtrantools}
\usepackage{epsfig}
\usepackage{amsmath, amssymb, amsbsy}
\usepackage{xspace}          
\usepackage[noend]{algpseudocode}
\usepackage{algorithm}
\usepackage{algorithmicx}
\usepackage{color}
\usepackage{cite}            
\usepackage{booktabs} 
\usepackage{verbatim} 
\usepackage[colorlinks=true,
            pdfauthor={Sven Puchinger, Johan Rosenkilde, Wenhui Li, Vladimir Sidorenko},
            pdftitle={Row Reduction Applied to Decoding of Rank-Metric and Subspace Codes},
            pdfproducer={Latex},
            pdfcreator={pdflatex}]
           {hyperref}
\usepackage{url}

\definecolor{navy}{rgb}{0,0,.5}

\makeatletter
\if@twocolumn
  \renewcommand\normalsize{%
   \@setfontsize\normalsize\@xpt{12.5pt}%
   \abovedisplayskip=3 mm plus6pt minus 4pt
   \belowdisplayskip=3 mm plus6pt minus 4pt
   \abovedisplayshortskip=0.0 mm plus6pt
   \belowdisplayshortskip=2 mm plus4pt minus 4pt
   \let\@listi\@listI}%

  \renewcommand\small{%
   \@setfontsize\small{8.5pt}\@xpt
   \abovedisplayskip 8.5\p@ \@plus3\p@ \@minus4\p@
   \abovedisplayshortskip \z@ \@plus2\p@
   \belowdisplayshortskip 4\p@ \@plus2\p@ \@minus2\p@
   \def\@listi{\leftmargin\leftmargini
               \parsep 0\p@ \@plus1\p@ \@minus\p@
               \topsep 4\p@ \@plus2\p@ \@minus4\p@
               \itemsep0\p@}%
   \belowdisplayskip \abovedisplayskip}

\else
  \if@smallext
   \renewcommand\normalsize{%
   \@setfontsize\normalsize\@xpt\@xiipt
   \abovedisplayskip=3 mm plus6pt minus 4pt
   \belowdisplayskip=3 mm plus6pt minus 4pt
   \abovedisplayshortskip=0.0 mm plus6pt
   \belowdisplayshortskip=2 mm plus4pt minus 4pt
   \let\@listi\@listI}%

  \renewcommand\small{%
   \@setfontsize\small\@viiipt{9.5pt}%
   \abovedisplayskip 8.5\p@ \@plus3\p@ \@minus4\p@
   \abovedisplayshortskip \z@ \@plus2\p@
   \belowdisplayshortskip 4\p@ \@plus2\p@ \@minus2\p@
   \def\@listi{\leftmargin\leftmargini
               \parsep 0\p@ \@plus1\p@ \@minus\p@
               \topsep 4\p@ \@plus2\p@ \@minus4\p@
               \itemsep0\p@}%
   \belowdisplayskip \abovedisplayskip}
 \else
  \renewcommand\normalsize{%
   \@setfontsize\normalsize{9.5pt}{11.5pt}%
   \abovedisplayskip=3 mm plus6pt minus 4pt
   \belowdisplayskip=3 mm plus6pt minus 4pt
   \abovedisplayshortskip=0.0 mm plus6pt
   \belowdisplayshortskip=2 mm plus4pt minus 4pt
   \let\@listi\@listI}%

  \renewcommand\small{%
   \@setfontsize\small\@viiipt{9.25pt}%
   \abovedisplayskip 8.5\p@ \@plus3\p@ \@minus4\p@
   \abovedisplayshortskip \z@ \@plus2\p@
   \belowdisplayshortskip 4\p@ \@plus2\p@ \@minus2\p@
   \def\@listi{\leftmargin\leftmargini
               \parsep 0\p@ \@plus1\p@ \@minus\p@
               \topsep 4\p@ \@plus2\p@ \@minus4\p@
               \itemsep0\p@}%
   \belowdisplayskip \abovedisplayskip}
  \fi
\fi
\let\footnotesize\small
\makeatother

\usepackage[final,tracking=true,kerning=true,spacing=true,factor=1100,stretch=10,shrink=20]{microtype}

\algrenewcommand\alglinenumber[1]{{\scriptsize#1}}
\algrenewcommand\algorithmicrequire{\textbf{Input:}}
\algrenewcommand\algorithmicensure{\textbf{Output:}}
\newcommand{\Ifline}[2]{\State \textbf{if }#1{ \textbf{then} }#2}

\usepackage{varioref}        
\usepackage{fancyref}

\makeatletter
\def\mkfancyprefix#1#2{%
\expandafter\def\csname fancyref#1labelprefix\endcsname{#1}%
\begingroup\def\x{\endgroup\frefformat{plain}}%
    \expandafter\x\csname fancyref#1labelprefix\endcsname
    {\MakeLowercase{#2}\fancyrefdefaultspacing##1}%
\begingroup\def\x{\endgroup\Frefformat{plain}}%
    \expandafter\x\csname fancyref#1labelprefix\endcsname
    {#2\fancyrefdefaultspacing##1}%
\begingroup\def\x{\endgroup\frefformat{vario}}%
    \expandafter\x\csname fancyref#1labelprefix\endcsname
    {\MakeLowercase{#2}\fancyrefdefaultspacing##1##3}%
\begingroup\def\x{\endgroup\Frefformat{vario}}%
    \expandafter\x\csname fancyref#1labelprefix\endcsname
    {#2\fancyrefdefaultspacing##1##3}%
}
\makeatother
\fancyrefchangeprefix{\fancyrefeqlabelprefix}{eqn}
\mkfancyprefix{ssec}{Section}
\mkfancyprefix{tbl}{Table}
\mkfancyprefix{thm}{Theorem}
\mkfancyprefix{lem}{Lemma}
\mkfancyprefix{cor}{Corollary}
\mkfancyprefix{prop}{Proposition}
\mkfancyprefix{prob}{Problem}
\mkfancyprefix{alg}{Algorithm}
\mkfancyprefix{inv}{Invariant}
\mkfancyprefix{ex}{Example}
\mkfancyprefix{line}{Line}
\mkfancyprefix{def}{Definition}
\mkfancyprefix{itm}{Item}
\mkfancyprefix{app}{Appendix}
\mkfancyprefix{rem}{Remark}
\newcommand{\cref}[1]{\Fref{#1}}

\def\ve#1{{\mathchoice{\mbox{\boldmath$\displaystyle #1$}}%
              {\mbox{\boldmath$\textstyle #1$}}%
              {\mbox{\boldmath$\scriptstyle #1$}}%
              {\mbox{\boldmath$\scriptscriptstyle #1$}}}}

\def\confType{Conference} 
\def\jourType{Journal}

\definecolor{orange}{rgb}{1,0.5,0}
\ifx\PubType\confType
  \newcommand{\todo}[1]{}

\else\ifx\PubType\jourType
  \newcommand{\todo}[1]{}

\else
  \newcommand{\todo}[1]{{\color{red}[#1]}}

\fi\fi

\newcommand{\param}{\mu}
\newcommand{\mo}{{-1}}
\newcommand{\K}{\F}

\newcommand{\F}{\mathbb{F}}
\newcommand{\R}{{\cal{R}}}
\newcommand{\Q}{{\cal{Q}}}

\newcommand{\M}{{\cal{M}}}
\newcommand{\V}{{\cal{V}}}
\newcommand{\MR}{\R^{m\times m}}

\newcommand{\OD}[1]{{\Delta(#1)}} 
\newcommand{\m}{\ensuremath{\ve{m}}}
\renewcommand{\u}{\ensuremath{\ve{u}}}
\renewcommand{\v}{\ensuremath{\ve{v}}}
\newcommand{\w}{\ve{w}}
\renewcommand{\vec}[1]{\ensuremath{\ve{#1}}}
\newcommand{\LPc}{\ensuremath{\mathrm{LP}}}
\newcommand{\LP}[2][\null]{\LPc_{#1}(#2)}
\newcommand{\LT}[1]{\ensuremath{\mathrm{LT}(#1)}}
\newcommand{\LC}[1]{\ensuremath{\mathrm{LC}(#1)}}

\newcommand{\word}[1]{\textnormal{#1}}
\newcommand\modop{\ \word{mod}\ }         
\newcommand{\code}[1]{\textup{\textsf{#1}}}
\newcommand{\NN}{\mathbb{Z}_{\geq 0}}
\newcommand{\N}{\mathbb{Z}_{> 0}}

\newcommand{\I}[1]{^{(#1)}}

\newcommand\vecVal{\psi}

\DeclareMathOperator{\maxdeg}{maxdeg}

\AtBeginDocument{}

\newcommand{\Fq}{\mathbb{F}_q}
\newcommand{\Fqs}{\mathbb{F}_{q^s}}
\newcommand{\ZZ}{\mathbb{Z}}
\newcommand{\QQ}{\mathbb{Q}}

\newcommand{\AP}{\mathcal{A}}

\newcommand{\mul}{\cdot}

\newcommand\ev{\word{ev}}

\newcommand{\omegaparamMV}{\chi}
\newcommand{\smallsum}{{\textstyle\sum\nolimits}}

\begin{document}

\title{Row Reduction Applied to Decoding of Rank-Metric and Subspace Codes\thanks{Sven Puchinger (grant BO 867/29-3), Wenhui Li and Vladimir Sidorenko (grant BO~867/34-1) were supported by the German Research Foundation ``Deutsche Forschungsgemeinschaft'' (DFG).
}}

\author{Sven Puchinger \and Johan Rosenkilde, n\'e Nielsen \and Wenhui Li \and Vladimir Sidorenko}
\authorrunning{Sven Puchinger, Johan Rosenkilde, n\'e Nielsen, Wenhui Li, Vladimir Sidorenko}
\institute{
Sven Puchinger \and Wenhui Li \at Institute of Communications Engineering, Ulm University, Germany \\ \email{sven.puchinger@uni-ulm.de, wenhui.li@uni-ulm.de}
\and
Johan Rosenkilde \at Department of Applied Mathematics and Computer Science, Technical University of Denmark \\
\email{jsrn@jsrn.dk}
\and
Vladimir Sidorenko \at Institute for Communications Engineering, TU München, Germany, and on leave from the Institute of Information Transmission Problems (IITP), Russian Academy of Sciences \\
\email{vladimir.sidorenko@tum.de}
}

\date{Received: date / Accepted: date}

\maketitle

\begin{abstract}
We show that decoding of $\ell$-Interleaved Gabidulin codes, as well as list-$\ell$ decoding of Mahdavifar--Vardy codes can be performed by row reducing skew polynomial matrices.
Inspired by row reduction of $\F[x]$ matrices, we develop a general and flexible approach of transforming matrices over skew polynomial rings into a certain reduced form.
We apply this to solve generalised shift register problems over skew polynomial rings which occur in decoding $\ell$-Interleaved Gabidulin codes.
We obtain an algorithm with complexity $O(\ell \param^2)$ where $\param$ measures the size of the input problem and is proportional to the code length $n$ in the case of decoding.
Further, we show how to perform the interpolation step of list-$\ell$-decoding Mahdavifar--Vardy codes in complexity $O(\ell n^2)$, where $n$ is the number of interpolation constraints.


\end{abstract}

\newpage

\section{Introduction}

Numerous recent publications have unified the core of various decoding algorithms for Reed--Solomon (RS) and Hermitian codes using row reduction of certain $\F[x]$-module bases.
First for the Guruswami--Sudan list decoder \cite{alekhnovich_linear_2005,lee_list_2008,cohn_ideal_2010}, then for Power decoding \cite{nielsen2013generalised,nielsen2014power} and also either type of decoder for Hermitian codes \cite{nielsen_sub-quadratic_2015}.
By factoring out coding theory from the core problem, we enable the immediate use of sophisticated algorithms developed by the computer algebra community such as \cite{giorgi_complexity_2003,zhou_efficient_2012}.

The goal of this paper is to explore the row reduction description over skew polynomial rings, with a main application for decoding rank-metric and subspace codes.
Concretely, we prove that Interleaved Gabidulin and Mahdavifar--Vardy codes can be decoded by transforming a module basis into weak Popov form, which can be obtained by a skew-analogue of the elegantly simple Mulders--Storjohann algorithm \cite{mulders2003lattice}.
By exploiting the structure of the module bases arising from the decoding problems, we refine the algorithm to obtain improved complexities.
These match the best known algorithms for these applications but solve more general problems, and it demonstrates that the row reduction methodology is both flexible and fast for skew polynomial rings.
Building on this paper, \cite{puchinger2016alekhnovich} proposes an algorithm which improves upon the best known complexity for decoding Interleaved Gabidulin codes.

\cref{ssec:relatedwork} summarizes related work.
We set basic notation in \cref{sec:prelim}.
\cref{sec:Decoding} shows how to solve the mentioned decoding problems using row reduction and states the final complexity results which are proven in the subsequent sections.
We describe row reduction of skew polynomial matrices in \cref{sec:RowReduction}.
\cref{sec:FasterRowRed} presents faster row reduction algorithms for certain input matrices, with applications to the decoding problems.

This work was partly presented at the International Workshop on Coding and Cryptography 2015 \cite{li2015solving}.
Compared to this previous work, we added the decoding of MV codes using the row reduction approach.\footnote{
  We opted for using the term ``row reduction'' rather than ``module minimisation'', as we used in \cite{li2015solving}, since the former is more common in the literature.}
It spurred a new refinement of the Mulders--Storjohann, in \cref{ssec:wpfwalk}, which could be of wider interest.

\subsection{Related Work}
\label{ssec:relatedwork}

In this paper we consider skew polynomial rings over finite fields without derivations \cite{ore1933theory} (see \cref{ssec:skewpolys} for this restricted definition of skew polynomials).
This is the most relevant case for coding theory, partly because they are easier to compute with, though non-zero derivations have been used in some constructions \cite{boucher2014linear}.
All of the row reduction algorithms in this paper work for skew polynomial rings with non-zero derivation, but the complexity would be worse.
The algorithms also apply to skew polynomial rings over any base ring, e.g.~$\K(z)$ or a number field, but due to coefficient growth in such settings, their bit-complexity would have to be analysed.

A skew polynomial ring over a finite field without derivation is isomorphic to a ring of \emph{linearised polynomials} under a trivial isomorphism, and the rings' evaluation maps agree.
Our algorithms could be phrased straightforwardly to work on modules over linearised polynomials.
Much literature on Gabidulin codes uses the language of linearised polynomials.

Skew polynomial rings are instances of Ore rings, and some previous work on computing matrix normal forms over Ore rings can be found in \cite{abramov_solutions_2001,beckermann2006fraction}.
The focus there is when the base ring is $\QQ$ or $\K(z)$ where coefficient growth is a major issue.
These algorithms are slower than ours when the base ring is a finite field.
\cite{middeke2011computational} considers a setting more similar to ours, but obtains a different algorithm and slightly worse complexity.

Gabidulin codes \cite{delsarte1978bilinear,gabidulin1985theory,roth1991maximum} are maximum rank distance codes over finite fields;
they are the rank-metric analogue of Reed--Solomon codes.
An Interleaved Gabidulin code is a direct sum of several Gabidulin codes, similar to Interleaved RS codes.
In a synchronised error model they allow an average-case error-correction capability far beyond half the minimum rank distance \cite{loidreau2006decoding}.
Decoding of Interleaved Gabidulin codes is often formulated as solving a simultaneous ``Key Equation'' \cite{sidorenko2014fast}.
Over $\K[x]$ this computational problem is also known as a multi-sequence shift-register synthesis \cite{feng_generalization_1991,sidorenko_linear_2011}, simultaneous Pad\'e approximation \cite{baker_pade_1996}, or vector rational function reconstruction \cite{olesh_vector_2006}.
This problem also has further generalisations, some of which have found applications in decoding of algebraic codes, e.g.~\cite{roth_efficient_2000,zeh_interpolation_2011,nielsen_sub-quadratic_2015}.
In the computer algebra community, Pad\'e approximations have been studied in a wide generality, e.g.~\cite{beckermann_uniform_1994};
to the best of our knowledge, analogous generalisations over skew polynomial rings have yet to see any applications.

Lately, there has been an interest in Gabidulin codes over number fields, with applications to space-time codes and low-rank matrix recovery \cite{augot_rank_2013}. 
Their decoding can also be reduced to a shift-register type problem \cite{mueelich2016alternative}, which could be solved using the algorithms in this paper (though again, one should analyse the bit-complexity).

Mahdavifar--Vardy (MV) codes \cite{mahdavifar2013algebraic,mahdavifar2012list} are subspace codes whose main interest lie in their
property of being list-decodable beyond half the minimum distance.
Their rate unfortunately tend to zero for increasing code lengths.
In \cite{mahdavifar_algebraic_2011}, Mahdavifar and Vardy presented a refined construction which can be decoded ``with multiplicities'' allowing a better decoding radius and rate; it is future work to adapt our algorithm to this case.
The decoding of MV codes is is heavily inspired by the Guruswami--Sudan algorithm for Reed--Solomon codes \cite{guruswami_improved_1999}, and our row reduction approach in \cref{sec:mv_codes} is similarly inspired by fast module-based algorithms for realising the Guruswami--Sudan \cite{lee_list_2008,beelen_key_2010}.

Another family of rank-metric codes which can be decoded beyond half the minimum distance are Guruswami--Xing codes \cite{guruswami_list_2013}.
These can be seen simply as heavily punctured Gabidulin codes, and their decoding as a \emph{virtual interleaving} of multiple Gabidulin codes.
This leads to a decoder based on solving a simultaneous shift-register equation, where our algorithms apply.
Guruswami--Wang codes \cite{guruswami_explicit_2014} are Guruswami--Xing codes with a restricted message set, so the same decoding algorithm applies.

Over $\K[x]$, row reduction, and the related concept of order bases, have been widely studied and sophisticated algorithms have emerged, e.g.~\cite{giorgi_complexity_2003,alekhnovich_linear_2005,zhou_efficient_2012}.
As a follow-up to this work, a skew-analogue of the algorithm in \cite{alekhnovich_linear_2005} was proposed in \cite{puchinger2016alekhnovich}.

\section{Preliminaries}\label{sec:prelim}

\subsection{Skew Polynomials}
\label{ssec:skewpolys}

Let $\K$ be a finite field and $\theta$ an $\K$-automorphism.
Denote by $\R=\K[x;\theta]$ the non-commutative ring of \emph{skew polynomials} over $\K$ (with zero derivation):
elements of $\R$ are of the form $\sum_i a_i x^i$ with $a_i \in \K$, addition is as usual, while multiplication is defined by $xa = \theta(a) x$ for all $a \in \K$.
When we say ``polynomial'', we will mean elements of $\R$.
The definition of the degree of a polynomial is the same as for ordinary polynomials.
See \cite{ore1933theory} for more details.

The evaluation map of $a \in \R$ is given as:
\begin{align*}
  a(\cdot) := \ev_a(\cdot) \, : \, & \K \to \K \ \\
  & \alpha \mapsto \smallsum_{i} a_i \theta^i(\alpha).
\end{align*}
This is a group homomorphism on $(\K,+)$, and it is a linear map over the fixed field of $\theta$.
Furthermore, for two $a,b \in \R$ we have $\ev_{ab} = \ev_a \circ \ev_b$.
This is sometimes known as operator evaluation, e.g.~\cite{boucher2014linear}.

If $\F_{q}$ is the field fixed by $\theta$ for some prime power $q$, then $\K = \F_{q^s}, s \in \N$, and $\theta(a) = a^{q^i}$ for some $0 \leq i < s$, i.e.~a power of the Frobenius automorphism of $\F_{q^s}/\;\F_q$.

\begin{definition}
For $a,b,c \in \R$, we write $a \equiv b \mod c$ (\emph{right modulo operation}) if there exists $d \in \R$ such that $a = b + d c$
\end{definition}

In complexity estimates we count the total number of the following operations: $+, -, \cdot, \slash$ and $\theta^i$ for any $i \in \N$.
For computing $\theta^i$ the assumption is that Frobenius automorphism can be done efficiently in $\Fqs$; this is reasonable since we can represent $\Fqs$-elements using a normal basis over $\Fq$ (cf. \cite[Section 2.1.2]{wachter-zeh_decoding_2013}): in this case, $a^q$ for $a \in \Fqs$ is simply the cyclic shift of $a$ represented as an $\Fq$-vector over the normal basis.

\subsection{Skew Polynomial Matrices}
\label{ssec:skewmatrices}

Free modules and matrices over $\R$ behave quite similarly to the $\K[x]$ case, keeping non-commutativity in mind:
\begin{itemize}
  \item Any left sub-module $\V$ of $\R^m$ is free and admits a basis of at most $m$ elements. 
  Any two bases of $\V$ have the same number of elements.
  \item The rank of a matrix $M$ over $\R$ is defined as the number of elements in any basis of the left $\R$-row space of $M$.
  The rows of two such matrices $M, M' \in \R^{n \times m}$ generate the same left module if and only if there exists a $U \in \mathrm{GL}_n(\R)$ such that $M = UM'$, where $\mathrm{GL}_n(\R)$ denotes the set of invertible $n \times n$ matrices over $\R$.
\end{itemize}

These properties follow principally from $\R$ being an Ore ring and therefore left Euclidean, hence left PID, hence left Noetherian\footnote{%
  $\R$ is also right Euclidean, a right PID and right Noetherian, but we will only need its left module structure.
}.
Moreover, $\R$ has a unique left skew field\footnote{%
  Skew fields are sometimes known as ``division rings''.
}
of fractions $\Q$ from which it inherits its linear algebra properties.
See e.g.~\cite{draxl1983skew,clark2012noncom} for more details.
In this paper we exclusively use the left module structure of $\R$, and we will often omit the ``left'' denotation.

We introduce the following notation for vectors and matrices over $\R$:
Matrices are denoted by capital letters (e.g. $V$). The $i$th row of $V$ is denoted by $\v_i$, the $j$th element of a vector $\v$ is $v_j$ and $v_{i,j}$ is the $(i,j)$th entry of a matrix $V$. Indices start at $0$.
\begin{itemize}\setlength\itemsep{.3em}
\item The \emph{degree of a vector} $\v$ is $\deg \v := \max_{i} \{ \deg v_i \}$ (and $\deg \ve 0 = -\infty$) and the \emph{degree of a matrix} $V$ is $\deg V := \sum_{i} \{ \deg \v_i \}$.
\item The \emph{max-degree} of $V$ is $\maxdeg V := \max_{i} \{ \deg \v_i \} = \max_{i,j} \{ \deg v_{i,j} \}$.
\item The \emph{leading position} of a non-zero vector $\v$ is $\LP{\v} := \max\{i \, : \, \deg v_i = \deg \v \}$, i.e. the \emph{rightmost} position having maximal degree in the vector.
      Furthermore, we define the \emph{leading term} $\LT{\v} := v_{\LP{\v}}$ and $\LC \v$ is the leading coefficient of $\LT \v$.
\end{itemize}

\subsection{The weak Popov form}
\label{ssec:wpf}

\begin{definition}
A matrix $V$ over $\R$ is in \emph{weak Popov form} if the leading positions of  all its non-zero rows are different.
\end{definition}

The following lemma describes that the rows of a matrix in weak Popov form are minimal in a certain way.
Its proof is exactly the same as for $\F[x]$ modules and is therefore omitted, see e.g.~\cite{nielsen2013generalised}.

\begin{lemma}\label{lem:minpop}
Let $V$ be a matrix in weak Popov form, and let $\V$ be the $\R$-module generated by its rows.
Then the non-zero rows of $V$ are a basis of $\V$ and every $\u \in \V$ satisfies $\deg \u \ge \deg \v$, where $\v$ is the row of $V$ with $\LP{\v}=\LP{\u}$.
\end{lemma}

We will need to ``shift'' the relative importance of some columns compared to others.
Given a ``shift vector'' $\w = (w_0,\dots,w_\ell) \in \NN^{\ell+1}$, define the mapping
\begin{align*}
\Phi_\w \, : \, \R^{\ell+1} \to \R^{\ell+1}, \; \u = (u_0,\dots,u_\ell) \mapsto (u_0 x^{w_0}, \dots, u_\ell x^{w_\ell}).
\end{align*}
It is easy to compute the inverse of $\Phi_\w$ for any vector in $\Phi_\w(\R^{\ell+1})$.
Note that since the monomials $x^{w_i}$ are multiplied from the right, applying $\Phi_\w$ will only \emph{shift} the entry polynomials, and not modify the coefficients.
We can extend $\Phi_\w$ to $\R$-matrices by applying it row-wise.
\begin{definition}
For any $\w = (w_0,\dots,w_\ell) \in \NN^{\ell+1}$, a matrix $V \in \R^{\cdot \times (\ell+1)}$ is in \emph{$\w$-shifted weak Popov} form if $\Phi_\w(V)$ is in weak Popov form.
\end{definition}

Given some matrix $V$ over $\R$, ``transforming $V$ into ($\w$-shifted) weak Popov form'' means to find some $W$ generating the same row space as $V$ and such that $W$ is in ($\w$-shifted) weak Popov form.
We will see in \cref{ssec:MS} that such $W$ always exist.

Throughout this paper, by ``row reduced'' we mean ``in weak Popov form''\footnote{%
There is a precise notion of ``row reduced'' \cite{kailath_linear_1980}[p. 384] for $\F[x]$ matrices.
Weak Popov form implies being row reduced, but we will not formally define row reduced in this paper.}.
Similarly, ``row reduction'' means ``transforming into weak Popov form''.

\section{Decoding Problems in Rank-Metric and Subspace Codes}
\label{sec:Decoding}

\subsection{Interlaved Gabidulin Codes: Multi-sequence shift registers}\label{sec:gab_codes}

It is classical to decode errors in a Gabidulin code by solving a syndrome-based ``Key Equation'': that is, a shift-register synthesis problem over $\R$, see e.g.~\cite{gabidulin1985theory}.
An Interleaved Gabidulin code is a direct sum of several Gabidulin codes \cite{loidreau2006decoding}, and error-decoding can be formulated as a shift-register synthesis of several sequences simultaneously.
A slightly more general notion of shift-register synthesis allows formulating the decoder using the ``Gao Key Equation'' \cite{wachter-zeh_decoding_2013}.
Another generalisation accommodates error-and-erasure decoding of some Gabidulin resp. Interleaved Gabidulin codes \cite{li2014transform,wachter-zeh_decoding_2013}.

All these approaches are instances of the following ``Multi-Sequence generalised Linear Skew-Feedback Shift Register'' (MgLSSR) synthesis problem:
\begin{problem}[MgLSSR]\label{prob:skewpade}
Given skew polynomials $s_i, g_i \in \R$ and non-negative integers $\gamma_i \in \NN$ for $i=1,\ldots,\ell$, find skew polynomials $\lambda, \omega_1,\dots,\omega_\ell \in \R$, with $\lambda$ of minimal degree such that the following holds:
\begin{IEEEeqnarray*}{rCl}
\lambda s_i &\equiv& \omega_i \mod g_i   \\
\deg \lambda + \gamma_0  &>& \deg\omega_i + \gamma_i
\end{IEEEeqnarray*}
\end{problem}

We show how to solve this problem by row reduction of a particular module basis.
The approach is analogous to how the $\K[x]$-version of the problem is handled by Rosenkilde in \cite{nielsen2013generalised}, with only a few technical differences due to the non-commutativity of $\R$.

In the sequel we consider a particular instance of \cref{prob:skewpade}, so $\R$, $\ell \in \N$, and $s_i, g_i \in \R$, $\gamma_i \in \NN$ for $i=1,\ldots,\ell$ are arbitrary but fixed.
We assume $\deg s_i \leq \deg g_i$ for all $i$ since taking $s_i := (s_i \modop g_i)$ yields the same solutions to \cref{prob:skewpade}.

Denote by $\M$ the set of all vectors $\v \in \R^{\ell+1}$ satisfying the congruence relation, i.e.,
\begin{equation}\label{eq:module_vectors}
\M := \big\{ (\lambda,\omega_1,\dots,\omega_\ell)\in \R^{\ell+1} \; \big| \; \lambda s_i \equiv \omega_i \mod g_i \; \forall i=1,\dots,\ell \big\}.
\end{equation}
\vspace*{-1.8em}
\begin{lemma}
  \label{lem:basis}
    Consider an instance of \cref{prob:skewpade} and $\M$ as in \eqref{eq:module_vectors}.
    $\M$ with component-wise addition and left multiplication by elements of $\R$ forms a free left module over $\R$.
    The rows of $M$ form a basis of $\M$, where
    \renewcommand{\arraystretch}{0.9}
    \begin{eqnarray}
      \label{eqn:M}
    M=\left (
        \begin{array}{ccccc}
        1 & s_1 & s_2 & \dots & s_\ell\\
        0 & g_1 & 0 & \dots  & 0 \\
        0 & 0 & g_2 & \dots  & 0 \\
        \vdots & \vdots & \ddots & \ddots&\vdots \\
        0 & 0 & 0 & \dots  & g_\ell \\
        \end{array}\right) \ . 
    \end{eqnarray}
\end{lemma}

\begin{proof}
Since $\M \subseteq \R^{\ell+1}$, the first half of the statement follows easily since $\M$ is clearly closed under addition and left $\R$ multiplication.
$M$ is a basis of $\M$ by arguments analogous to the $\F[x]$ case, cf. \cite[Lemma~1]{nielsen2013generalised}.
\end{proof}

\cref{lem:basis} gives a simple description of all solutions of the congruence requirement of \cref{prob:skewpade} in the form of the row space of an explicit matrix $M$.
The following theorem implies that computing the weak Popov form of $M$ is enough to solve \cref{prob:skewpade}.
The proof is similar to the $\K[x]$-case but since there is no convenient reference for it, we give it here for the $\R$ case.
The entire strategy is formalised in \cref{alg:shiftregisteralg}.

\begin{theorem}\label{thm:gab_rowred}
  Consider an instance of \cref{prob:skewpade}, and $\M$ as in \eqref{eq:module_vectors}.
  Let $\w = (\gamma_0,\dots,\gamma_\ell) \in \NN^{\ell+1}$.
If $V$ is a basis of $\M$ in $\w$-shifted weak Popov form, the row $\v$ of $V$ with $\LP{\Phi_\w(\v)}=0$ is a solution to \cref{prob:skewpade}.
\end{theorem}

\begin{proof}
  By \cref{lem:basis} the row $\v$ satisfies the congruence requirement of \cref{prob:skewpade}.
  For the degree restriction of \cref{prob:skewpade}, note that any $\u \in \M$ satisfies this restriction if and only if $\LP{\Phi_\w(\u)} = 0$, since $\deg u_i + \gamma_i = \deg(\Phi_\w(\u)_i)$.
  Furthermore, if this is the case, then $\deg(\Phi_\w(\u)) = \deg u_0 + \gamma_0$.
  Thus, not only must $\v$ satisfy the degree restriction, but by \cref{lem:minpop}, $v_0$ also has minimal possible degree.
\end{proof}

\begin{algorithm}[H]
  \caption{Solve \cref{prob:skewpade} by Row Reduction}
  \label{alg:shiftregisteralg}
  \begin{algorithmic}[1]
    \Require{Instance of \cref{prob:skewpade}}.
    \Ensure{Solution $\v = (\lambda,\omega_1,\dots,\omega_\ell)$ of \cref{prob:skewpade}.}
    \State Set up $M$ as in \eqref{eqn:M}.
    \State Compute $V$ as a $\w$-shifted weak Popov form of $M$.
      \label{line:shiftregister:wpf}
    \State \Return the row $\v$ of $V$ having $\LP{\Phi_\w(\v)} = 0$.
  \end{algorithmic}
\end{algorithm}

The complexity of \cref{alg:shiftregisteralg} is determined by \cref{line:shiftregister:wpf}.
Therefore, in Sections~\ref{sec:RowReduction} and \ref{ssec:DD} we analyse how and in which complexity we can row-reduce $\R$-matrices.
In particular, we prove the following statement, where $\param := \max_i \{\gamma_i +\deg g_i\}$.

\begin{theorem}
  \label{thm:shiftregister_complexity}
\cref{alg:shiftregisteralg} has complexity $%
\begin{cases}
O(\ell \param^2), &\text{if } g_i = x^{t_i}+c_i, \, t_i \in \N, \, c_i \in \K \; \forall \, i, \\
O(\ell^2 \param^2), &\text{otherwise.}
\end{cases}$
\end{theorem}

\begin{proof}
The first case follows from \cref{thm:dd_fastonsimple} in \cref{ssec:DD}, using \cref{alg:dd} for the row reduction step.
For general $g_i$'s, the result of \cref{ex:MS_shift_register_input} in \cref{sec:RowReduction} holds, which estimates the complexity of \cref{alg:ms} for a shift-register input.
\end{proof}

The above theorem applies well to decoding Gabidulin and Interleaved Gabidulin codes since the $g_i$ are often in the restricted form: specifically, $g_i$ is a power of $x$ in syndrome Key Equations, while $g_i = x^n - 1$ in Gao Key Equation whenever $n \mid s$. 
We therefore achieve the same complexity as \cite{sidorenko2011skew} but in a wider setting.

\subsection{Decoding Mahdavifar--Vardy Codes}\label{sec:mv_codes}

Mahdavifar--Vardy (MV) codes \cite{mahdavifar2012list,mahdavifar2013algebraic} are subspace codes constructed by evaluating powers of skew polynomials at certain points.
We will describe how one can use row reduction to carry out the most computationally intensive step of the MV decoding algorithm given in \cite{mahdavifar2013algebraic}, the Interpolation step.
In this section, $\R = \Fqs[x;\theta]$ where $\theta$ is some power of the Frobenius automorphism of $\Fqs/\,\Fq$.

\begin{problem}[Interpolation Step of MV decoding]
  \label{prob:mv_interpolation}
  Let $\ell, k, s, n \in \N$ be such that $\binom{\ell+1} 2 (k-1) < n \leq s$.
  Given $(x_i, y_{i,1},\ldots, y_{i,\ell}) \in \Fqs^{\ell+1}$ for $i=1,\ldots,n$, where the $x_i$ are linearly independent over $\Fq$, 
  find a non-zero $\vec Q \in \R^{\ell+1}$ satisfying:
\begin{IEEEeqnarray}{rCl+l}
Q_0(x_i) + \sum_{t=1}^\ell Q_t(y_{i,t}) & = & 0                          & i=1,\ldots,n, \label{eq:interpol_zero} \\
\deg Q_t                          & < & \omegaparamMV - t(k-1)  & t=0,\dots,\ell, \label{eq:interpol_deg}
\end{IEEEeqnarray}
where $\omegaparamMV$ is given by
\begin{equation*}
 \omegaparamMV = \left\lceil \frac{n+1}{\ell+1} + \frac{1}{2}\ell(k-1) \right\rceil \\ 
\end{equation*}
\end{problem}

The problem can be solved by a large linear system of equations whose dimensions reveals that a solution always exists \cite[Lemma~8]{mahdavifar2013algebraic}.
Note that the requirement $n > \binom{\ell+1} 2 (k-1)$ ensures that all the degree bounds \eqref{eq:interpol_deg} are non-negative.

Let $\M$ be the set of all $\vec Q$ that satisfy \eqref{eq:interpol_zero} though not necessarily \eqref{eq:interpol_deg}:
\begin{equation}\label{eq:module_MV}
  \textstyle
  \M = \big\{ \vec Q \in \R^{\ell+1} \ \big|\ Q_0(x_i) + \sum_{t=1}^\ell Q_t(y_{i,t}) = 0 \quad i=1,\ldots,n \big\}
\end{equation}

\begin{lemma}
  \label{lem:q_module}
  Consider an instance of \cref{prob:mv_interpolation}. Then $\M$ of \eqref{eq:module_MV} is a left $\R$-module.
\end{lemma}

\begin{proof}
$\M$ is closed under addition since $a(\alpha) + b(\alpha) = (a+b)(\alpha)$ for all $a,b \in \R$ and $\alpha \in \Fqs$.
Let $f \in \R$, $\vec Q = (Q_0,Q_1,\dots,Q_\ell) \in \M$. Then $f \mul \vec Q$ satisfies \eqref{eq:interpol_zero} since
\begin{align*}
(f \mul Q_0)(x) + \sum_{i=1}^{\ell} (f \mul Q_i)(y_i) = f\Big(Q_0(x) + \sum_{i=1}^{\ell} Q_i(y_i)\Big) = f(0) = 0 \ .
  \\[-3em]
\end{align*}
\end{proof}

\noindent
For explicitly describing a basis of $\M$, we need a few well-known technical elements:
\begin{definition}
  Given $a_1,\ldots,a_m \in \Fqs$ which are linearly independent over $\Fq$, the \emph{annihilator polynomial} of the $a_i$ is the monic non-zero $\AP \in \R$ of minimal degree such that $\AP(a_i) = 0$ for all $i$.
\end{definition}

It is easy to show that the annihilator polynomial is well-defined and that $\deg \AP = m$, see e.g.~\cite{ore_special_1933}.
The existence of annihilator polynomials easily leads to the following analogue of Lagrange interpolation:
\begin{lemma}[Interpolation polynomial]
  Given any $a_1,\dots,a_m \in \Fqs$ which are linearly independent over $\Fq$, and arbitrary $b_1,\dots,b_m \in \Fqs$, there exists a unique $R \in \R$ of degree at most $m-1$ such that $R(a_i) = b_i$ for all $i=1,\ldots,m$.
\end{lemma}

\begin{lemma}
  \label{lem:mv_basis}
  Consider an instance of \cref{prob:mv_interpolation} and let $\M$ be as in \eqref{eq:module_MV}. 
  Denote by $G$ the annihilator polynomial of the $x_i, i=1,\ldots,n$, and let $R_t \in \R, t=1,\ldots,\ell$ be the interpolation polynomial with $R_t(x_i) = y_{i,t}$ for $i=1,\dots,n$.
  The rows of $M$ form a basis of $\M$:
\begin{align}
M =
\begin{pmatrix}
\m_0 \\
\m_1 \\
\m_2 \\
\vdots \\
\m_\ell
\end{pmatrix}
=
\begin{pmatrix}
G 		& 0 & 0 & \dots & 0 \\
-R_1	& 1	& 0 & \dots & 0 \\
-R_2	& 0 & 1	& \dots & 0 \\
\vdots	& \vdots & \vdots & \ddots & \vdots \\
-R_\ell	& 0 & 0 & \dots & 1
\end{pmatrix}
\label{eq:M_MV}
\end{align}
\end{lemma}

\begin{proof}
``$\subseteq$'':
We should show that each $\vec m_j$ all ``vanish'' at the points $(x_i,y_{i,1},\ldots,y_{i,\ell})$.
Consider such a point; we have two cases:
\begin{IEEEeqnarray*}{r'l'l}
\m_0: & G(x_i) = 0 \\
\m_t: & 1(y_{i,t})-R_t(x_i) = y_{i,t}-R_t(x_i) = 0 \ ,  & t=1,\dots,\ell
\end{IEEEeqnarray*}

``$\supseteq$'':
Consider some $\vec Q = (Q_0,\dots,Q_\ell) \in \M$.
Then we can write
\begin{IEEEeqnarray*}{rCrClCl}
  \v_\ell &:=& \vec Q \\
  \v_{\ell-1} &:=& \v_\ell &-& v_{\ell,\ell} \mul \m_\ell &=& (v_{\ell-1, 0},\ldots,v_{\ell-1,\ell-1}, 0) \\
  \v_{\ell-2} &:=& \v_{\ell-1} &-& v_{\ell-1,\ell-1} \mul \m_{\ell-1} &=& (v_{\ell-2,0}, \ldots, v_{\ell-2,\ell-2},0, 0) \\
  & \; \; \, \vdots \\
  \v_{0} &:=& \v_1 &-& v_{1,1} \mul \m_1 &=& (v_{0,0}, 0, \dots, 0).
\end{IEEEeqnarray*}
Since $\v_\ell \in \M$, and each $\m_t \in \M$, we conclude that all the $\v_t \in \M$ and in particular $\v_0 \in \M$.
Thus for any $i$ we must have $v_{0,0}(x_i) = 0$.
This means $G$ must right-divide $v_{0,0}$: for otherwise, the division would yield a non-zero remainder $B \in \R$ with $\deg B < \deg G$ but still having $B(x_i) = 0$, contradicting the minimality of $G$.

Summarily, $\v_0 = f\mul\m_0$ for some $f \in \R$, and hence $\vec Q = \v_\ell$ is an $\R$-linear combination of the rows of $M$.
\end{proof}

To complete the interpolation step, we need to find an element of $\M$ whose components satisfy the degree constraints \eqref{eq:interpol_deg}.

\begin{theorem}
  \label{thm:mv_wpfsolves}
  Consider an instance of \cref{prob:mv_interpolation}, and let $\M$ be as in \eqref{eq:module_MV}.
  Let $\w = (0,(k-1), \dots, \ell(k-1))$, and $V$ be a basis of $\M$ in $\w$-shifted weak Popov form.
  If $\v$ is a row of $V$ with minimal $\w$-shifted degree, $\deg \Phi_\w(\v)$, then $\vec v$ is a
  solution to \cref{prob:mv_interpolation}.
\end{theorem}

\begin{proof}
Any row of $V$ satisfies \eqref{eq:interpol_zero} because it is in $\M$.
As previously remarked, there exists some solution $\vec Q = (Q_0, Q_1,\dots,Q_\ell) \in \M$ satisfying the degree conditions \eqref{eq:interpol_deg}.
By the choice of $\v$ and by \Fref[vario]{lem:minpop}, then $\deg \Phi_\w(\v) \leq \deg \Phi_\w(\vec Q)$.
But then if $t = \LP{\Phi_\w(\vec Q)}$ we have that for any $i$:
\begin{align*}
  &\deg \left(v_i x^{i(k-1)}\right) \leq \deg \Phi_\w(\vec Q) = \deg\left(Q_t x^{t(k-1)}\right) < \chi
\end{align*}
Hence, $\v$ satisfies \eqref{eq:interpol_deg}.
\end{proof}

This results immediately in the decoding procedure outlined as \cref{alg:mv_interpol}.

\begin{algorithm}[H]
  \caption{MV Interpolation Step by Row Reduction}
  \label{alg:mv_interpol}
  \begin{algorithmic}[1]
    \Require{An instance of \cref{prob:mv_interpolation}}
    \Ensure{A vector $\vec Q \in \R^{\ell+1}$ solving \cref{prob:mv_interpolation}.}
    \State Set up $M$ as in \eqref{eq:M_MV}.
    \State Compute a $\w$-shifted weak Popov form $V$ of $M$.
      \label{line:mv_interpol}
    \State \Return the row $\v$ of $V$ which has minimal $\w$-shifted degree $\deg \Phi_\w(\v)$.
  \end{algorithmic}
\end{algorithm}

\begin{theorem}
  \label{thm:mv_complexity}
\cref{alg:mv_interpol} has complexity $O(\ell n^2)$ over $\Fqs$.
\end{theorem}
\begin{proof}
  Computing $G$ can be done straightforwardly in $O(n^2)$ operations over $\Fqs$.
  Each $R_t$ can be computed in the same speed using a decomposition into smaller interpolations and two annihilator polynomials, see e.g.~\cite{puchinger_fast_2015}.
  For \cref{line:mv_interpol}, we use \cref{alg:mv_wpfwalk} whose complexity is $O(\ell n^2)$, proved as \cref{thm:wpf_correctness}. 
\end{proof}

In \cite{xie_linearized_2013}, Xie, Lin, Yan and Suter present an algorithm for solving the Interpolation Step using a skew-variant of the K\"otter--Nielsen--H\o{}holdt algorithm \cite{nielsen_decoding_1998} with complexity $O(\ell^2 s n)$ over $\Fqs$.
Since $n < s$, our algorithm is at least as fast as theirs.
Note that these costs probably dominate the complexity of MV decoding: the other step, Root-finding, likely\footnote{%
  In \cite{mahdavifar2013algebraic}, the claimed complexity of their root-finding is $O(\ell^{O(1)} k)$.
  However, we have to point out that the complexity analysis of that algorithm has severe issues which are outside the scope of this paper to amend.
  There are two problems: 1) It is not proven that the recursive calls will not produce many spurious ``pseudo-roots'' which are sifted away only at the leaf of the recursions; and 2) the cost analysis ignores the cost of computing the shifts $Q(X, Y^q + \gamma Y)$.
  Issue 1 is necessary to resolve for assuring polynomial complexity.
  For the original $\K[x]$-algorithm this is proved as \cite[Proposition 6.4]{roth_efficient_2000}, and an analogous proof might carry over.
  Issue 2 is critical since these shifts dominate the complexity: assuming the algorithm makes a total of $O(\ell k)$ recursive calls to itself, then $O(\ell k)$ shifts need to be computed, each of which costs $O(\ell \deg_x Q) \subset O(\ell n)$.
  If Issue 1 is resolved the algorithm should then have complexity $O(\ell^2 k n)$.
  }
has complexity $O(\ell^2 k n)$.

\section{Row Reduction of $\R$-matrices}
\label{sec:RowReduction}

\subsection{The Mulders--Storjohann Algorithm}\label{ssec:MS}

In this section, we introduce our algorithmic work horse: obtaining row reduced bases of left $\R$-modules $\V \subseteq \R^{m}$.
The core is an $\R$-variant of the Mulders--Storjohann algorithm \cite{mulders2003lattice} that was originally described for $\K[x]$ matrices.
The algorithm and its proof of correctness carries over almost unchanged, while a fine-grained complexity analysis is considerably more involved; we return to this in \cref{ssec:mscomplexity}.

\begin{definition}\label{def:simple_transformation}
  Applying a \emph{simple transformation $i$ on $j$ at position $h$} on a matrix $V$ with $\deg v_{i,h} \le \deg v_{j,h}$ means to replace $\v_j$ by $\v_j - \alpha x^\beta \v_i$,
  where $\beta = \deg v_{j,h} - \deg v_{i,h}$ and $\alpha = \LC{v_{j,h}}/\theta^{\beta}(\LC{v_{i,h}})$.

  By a \emph{simple LP-transformation $i$ on $j$}, where $\LP{\v_i} = \LP{\v_j}$, we will mean a simple transformation $i$ on $j$ at position $\LP{\v_i}$.
\end{definition}

\begin{remark}
Note that a simple transformation $i$ on $j$ at position $h$ cancels the leading term of the polynomial $v_{j,h}$.
Elementary row operations keep the row space and rank of the matrix unchanged, and in particular so does any sequence of simple transformations.
\end{remark}

\noindent We use the following \emph{value function} for $\R$ vectors as a ``size'' of $\R^{m}$ vectors:
\begin{IEEEeqnarray*}{r;cCl}
  \vecVal: & \R^{m} & \to     & \NN \\
           & \v          & \mapsto & \left\{
         \begin{array}{l@{\quad}l}
           0 & \textrm{ if } \v = \vec 0 \\
           m\deg \v + \LP{\v} + 1 & \textrm{otherwise}
         \end{array} \right.
\end{IEEEeqnarray*}

\begin{lemma}\label{lem:val_dec}
  For some $V \in \R^{\cdot \times m}$, consider a simple LP-transformation $i$ on $j$, where $\v_j$ is replaced by $\v'_j$.
  Then $\vecVal(\v_j')<\vecVal(\v_j)$.
\end{lemma}
\begin{proof}
The proof works exactly as in the $\F[x]$ case, cf. \cite[Lemma~8]{nielsen2013generalised}.
\end{proof}

\begin{algorithm}[t]
  \caption{Mulders--Storjohann for $\R$ matrices}
  \label{alg:ms}
  \begin{algorithmic}[1]
    \Require{A matrix $V$ over $\R$, whose rows span a module $\V$.}
    \Ensure{A basis of $\V$ in weak Popov form.}
    \State Until no longer possible, apply a simple LP-transformation on two rows in $V$.
    \State \Return $V$.
  \end{algorithmic}
\end{algorithm}

\begin{theorem}
  \label{thm:ms}
\cref{alg:ms} is correct.
\end{theorem}
\begin{proof}
  By \cref{lem:val_dec}, the $\vecVal$-value of one row of $V$ decreases for each simple LP-transformation.
  The sum of the values of the rows must at all times be non-negative so the algorithm must terminate.
  When the algorithm terminates there are no $i \neq j$ such that $\LP{\v_i}= \LP{\v_j}$.
  That is to say, $V$ is in weak Popov form.
\end{proof}

The above proof easily leads to the rough complexity estimate of \cref{alg:ms} of $O(m^2 \deg V \maxdeg V)$, where $m$ is the number of columns in $V$.

Note that in \cref{alg:ms} each iteration might present several possibilities for the simple LP-transformation; the above theorem shows that any choice of LP-transformations leads to the correct result.

To transform $V$ into $\w$-shifted weak Popov form, for some shift $\w \in \NN^m$, we let $V' = \Phi_\w(V)$ and apply \cref{alg:ms} on $V'$ to obtain $W'$ in weak Popov form.
Since \cref{alg:ms} only performs row operations, it is clear that $\Phi_\w$ can be inverted on $W'$ to obtain $W = \Phi_\w^\mo(W')$.
Then $W$ is in $\w$-shifted weak Popov form by definition.

\subsection{The Determinant Degree and Orthogonality Defect}\label{sec:complexity}

The purpose of this section is to introduce the orthogonality defect as a tool for measuring ``how far'' a square, full-rank matrix over $\R$ is from being in weak Popov form.
It relies on the nice properties of the degree of the Dieudonné determinant for matrices over $\R$.
The orthogonality defect for $\K[x]$ matrices was introduced by Lenstra \cite{lenstra85factoring} and used in \cite{nielsen2013generalised} to similar effect as we do here.

Dieudonné introduced a function for matrices over skew fields which shares some of the essential properties of the usual commutative determinant, in particular that it is multiplicative, see \cite{dieudonne1943det} or \cite[\S 20]{draxl1983skew}.
This Dieudonné determinant can be applied to matrices over $\R$ by considering $\R$ inside its left field of fractions.
The definition of this determinant is quite technical, and we will not actually need to invoke it.
Rather, we will use an observation by Taelman \cite{taelman2006dieudonne} that the Dieudonné determinant implies a simple-behaving \emph{determinant degree} function for matrices with very nice properties:
\begin{proposition}
  \label{prop:degdet}
  There is a unique function $\deg\det:\,\MR \to \ZZ_{\geq 0} \cup \{-\infty\}$ s.t.:
  \begin{itemize}
    \setlength{\itemsep}{0.2em}
    \item $\deg\det(A A') = \deg\det(A) + \deg\det(A')$ for all $A, A' \in \MR$.
    \item $\deg\det U = 0$ for all $U \in GL_m(\R)$.
    \item If $A$ is diagonal with diagonal elements $d_0,\ldots,d_{m-1}$, then $\deg \det A = \sum_{i}\deg d_i$
  \end{itemize}
\end{proposition}

\begin{corollary}
  \label{cor:degdet}
  For any $A,A' \in \MR$ then:
\begin{itemize}
  \item If $A'$ is obtained from $A$ by elementary row operations, then $\deg\det A' = \deg\det A$.
  \item If $B$ equals $A \in \MR$ with one row or column scaled by some $f \in \R^*$, then $\deg\det B = \deg f + \deg \det A$.
  \item If $A$ is triangular with diagonal elements $d_0,\ldots,d_{m-1}$, then $\deg\det A = \sum_i\deg d_i$.
  \item $\deg\det(\Phi_\w(A)) = \deg\det(A) + \sum_i w_i$ for any shift $\vec w$.
\end{itemize}
\end{corollary}

\begin{example}\label{ex:detdeg}
Consider input matrix $\Phi_\w(M)$ to \cref{alg:shiftregisteralg} for the case\footnote{This is a realistic shift register problem arising in decoding of an Interleaved Gabidulin code with $n=s=100$, $k_1 = 58$, $k_2 = 31$.} $\ell=2$, $\w = (\gamma_0, \gamma_1, \gamma_2) = (100, 42,69)$, $\deg s_1 = 99$, $\deg s_2 = 95$ and $\deg g_1 = \deg g_2 = 100$.
Then
\begin{align*}
  \arraycolsep=5pt
  \Phi_\w(M) = 
  \begin{pmatrix}
  x^{\gamma_0} & s_1 x^{\gamma_1} & s_2 x^{\gamma_2} \\
               & g_1 x^{\gamma_1} &                  \\
               &                  & g_2 x^{\gamma_2}
  \end{pmatrix}
  = 
  \begin{pmatrix}
  1 & s_1 & s_2 \\
    & g_1 &     \\
    &     & g_2 
  \end{pmatrix}
  \begin{pmatrix}
       x^{\gamma_0} \\
    &  x^{\gamma_1} \\
    && x^{\gamma_2}
  \end{pmatrix}
    \ .
\end{align*}
And so by \cref{prop:degdet}, $\deg\det \Phi_\w(M) = \deg g_1 + \deg g_2 + \sum_i \gamma_i = 411$. \hfill \IEEEQEDhere
\end{example}

This description of $\deg\det(\cdot)$ is not operational in the sense that it is not clear how to compute $\deg\det V$ for general $V \in \MR$.
The following definition and \cref{prop:ODzero} implies that \cref{alg:ms} can be used to compute $\deg\det V$; conversely, we show in \cref{ssec:mscomplexity} how to bound the complexity of \cref{alg:ms} based on $\deg\det V$.

\begin{definition}
  The orthogonality defect of $V \in \MR$ is $\OD{V} := \deg V - \deg \det V$.
\end{definition}

\noindent
The following observations are easy for $\K[x]$ matrices, but require more work over $\R$:

\begin{proposition}\label{prop:ODzero}
Let $V \in \R^{m \times m}$ of full rank and in weak Popov form. Then $\OD{V} = 0$.
\end{proposition}
\begin{proof}
  Due to \cref{cor:degdet}, we can assume the columns and rows of $V$ are ordered such that $\LP{\vec v_i} = i$ and $\deg v_{i,i} \leq \deg v_{j,j}$ for $i < j$.
  We will call this property ``ordered weak Popov form'' in this proof.
  Note that it implies $\vecVal(\vec v_i) < \vecVal(\vec v_j)$ for $i < j$.
  We will inductively obtain a series of matrices $V\I 0 = V, V\I 1, V\I 2, \ldots, V\I m$ each in ordered weak Popov form, and such that the first $i$ columns of $V\I i$ are zero below the diagonal.
  Then $V\I m$ is upper triangular and we can obtain two expressions for its $\deg\det$.

  So assume that $V\I i$ is in ordered weak Popov form and its first $i$ columns are zero below the diagonal.
  Recall that the (left) \emph{union} of two skew polynomials $f, g \in \R$ is the unique lowest-degree $p \in \R$ such that $p = af = bg$ for some $a,b \in \R$; it is a consequence of the Euclidean algorithm that the union always exists, see e.g.~\cite{ore1933theory}.
  For each $j > i$ consider now the coefficients in the union of $v\I i_{i,i}$ and $v\I i_{j,i}$, i.e.~$a\I i_j,b\I i_j \in \R$ such that $a\I i_j v\I i_{i,i} = b\I i_j v\I i_{j,i}$.
  Let
  \[
    V\I {i+1} =
    \left(\begin{array}{@{}c|cccc@{}}
       \ I_{i-1}\  &   \\\hline
        & 1 \\
        & -a\I i_{i+1} & b\I i_{i+1} \\
        & \vdots  & & \ddots \\
        & -a\I i_{m-1} & & & b\I i_{m-1}
    \end{array}\right)
    V\I i \ ,
  \]
  where $I_{i-1}$ is the $(i-1) \times (i-1)$ identity matrix.
  The $i+1$ first columns of $V\I{i+1}$ are then zero below the diagonal.
  Also $\LP{a\I i_j\vec v\I i_i} < \LP{b\I i_j\vec v\I i_{j}} = \LP{\vec v\I i_j}$ and $\deg(a\I i_j\vec v\I i_i) \leq \deg(b\I i_j\vec v\I i_{j})$ for $j > i$, which means $\vecVal(a\I i_j\vec v\I i_i) < \vecVal(b\I i_j\vec v\I i_{j})$ and therefore $\vecVal(\vec v\I{i+1}_j) = \vecVal(b\I i_j\vec v\I i_{j})$.
  This implies that $V\I{i+1}$ is in ordered weak Popov form and that $\deg v\I{i+1}_{j,j} = \deg b\I i_j + \deg v\I i_{j,j}$ for $j > i$, which inductively expands to
  \[
    \deg v\I{i+1}_{j,j} = \deg v_{j,j} + \sum_{h=0}^i \deg b\I h_j \ .
  \]
  Inductively, we therefore arrive at an upper triangular matrix $V\I m$ in ordered weak Popov form, and whose diagonal elements satisfy $\deg v\I m_{j,j} = \deg v_{j,j} + \sum_{i=0}^{j-1} \deg b\I i_j$.
  Thus $\deg\det V\I m$ is the sum of all these degrees by \cref{cor:degdet}.
  On the other hand $V\I m$ is obtained by multiplying triangular matrices on $V\I 0 = V$, so by \cref{prop:degdet} we get another expression for $\deg\det V\I m$ as:
  \[
    \deg\det V\I m = \deg\det V + \sum_{i=0}^{m-1} \sum_{j=i+1}^{m-1} \deg b\I i_j
  \]
  Combining the expressions, we get $\deg\det V = \sum_{i=0}^{m-1} \deg v_{i,i} = \deg V$.
\end{proof}

\begin{corollary}\label{cor:degdet_is_small}
  Let $V \in \MR$ and full-rank, then $0 \leq \deg \det V \leq \deg V$.
\end{corollary}
\begin{proof}
  Applying \cref{alg:ms} on $V$ would use row operations to obtain a matrix $V' \in \MR$ in weak Popov form. 
  Then $\deg \det V = \deg\det V'$ by \cref{prop:degdet}.
  By \cref{prop:ODzero} then $\deg \det V' = \deg V' \geq 0$, and by the nature of \cref{alg:ms}, then $\deg V' \leq \deg V$.
\end{proof}

\subsection{Complexity of Mulders--Storjohann}
\label{ssec:mscomplexity}

We can now bound the complexity of \cref{alg:ms} using arguments completely analogous to the $\K[x]$ case in \cite{nielsen2013generalised}.
These are in turn, the original arguments of \cite{mulders2003lattice} but finer grained by using the orthogonality defect.
We bring the full proof here since the main steps are referred to in \cref{sec:demanddriven}.

\begin{theorem}\label{thm:ms_complexity_general}
\cref{alg:ms} with a full-rank input matrix $V \in \R^{m \times m}$ performs at most $m \big(\OD{V} + m \big)$ simple LP-transformations, and it has complexity $O(m^2 \OD{V} \maxdeg(V))$ over $\K$.
\end{theorem}

\begin{proof}
By \cref{lem:val_dec}, every simple LP-transformation reduces the $\vecVal$-value of one row with at least 1.
So the number of possible simple LP-transformations is upper bounded by the difference of values of the input matrix $V$ and the output matrix $U$, the matrices values being the sum of their rows'.
More precisely, the number of iterations is upper bounded by:
\begin{eqnarray*}
&&   {\textstyle \sum_{i=0}^{m-1}} [m\deg \v_i+\LP{\v_i}-\big( m\deg \u_i+\LP{\u_i} \big)] \\
&\leq& m^2 + m {\textstyle \sum_{i=0}^{m-1}} [\deg \v_i-\deg \u_i]  \\
&=& m [\deg V-\deg U+m] = m (\OD{V}+m) , 
\end{eqnarray*}
where the last equality follows from $\deg U = \deg \det U$ due to \cref{prop:ODzero} and $\deg \det U = \deg \det V$.

One simple transformation consists of calculating $\v_j - \alpha x^\beta \v_i$, so for every coefficient in $\v_i$, we must apply $\theta^\beta$, multiply by $\alpha$ and then add it to a coefficient in $\v_j$, each being in $O(1)$.
Since $\deg \v_j \leq \maxdeg(V)$ this costs $O(m\maxdeg(V))$ operations in $\K$.
\end{proof}

Since $\OD{V} \leq \deg V$, the above complexity bound is always at least as good as the straightforward bound we mentioned at the end of \cref{ssec:MS}.

\begin{example}[Mulders--Storjohann algorithm on an MgLSSR]
  \label{ex:MS_shift_register_input}
  Consider an instance of \cref{prob:skewpade}.
  The complexity of \cref{alg:shiftregisteralg} is determined by a row reduction of
\begin{align}
\arraycolsep=4pt
\Phi_\w(M) = \left (
        \begin{array}{ccccc}
        x^{\gamma_0} & s_1 x^{\gamma_1} & s_2 x^{\gamma_2} & \dots  & s_\ell x^{\gamma_\ell} \\
                     & g_1 x^{\gamma_1} &                                                    \\
                     &                  & g_2 x^{\gamma_2} &                                 \\
                     &                  &                  & \ddots &                        \\
                     &                  &                  &        & g_\ell x^{\gamma_\ell} \\
        \end{array}\right).
  \label{eqn:phiM}
\end{align}
Let $\param := \max_i \{\gamma_i +\deg g_i\}$.
We can assume that $\gamma_0 < \max_{i \geq 1} \{ \gamma_i + \deg s_i \} \leq \param$ since otherwise $M$ is already in $\w$-shifted weak Popov form.
To apply \cref{thm:ms_complexity_general}, we calculate the orthogonality defect of $\Phi_\w(M)$.
Since it is upper triangular, the degree of its determinant is
\begin{equation*}
  \textstyle
\deg \det \Phi_\w(M) =  \sum_{i=1}^\ell \deg g_i + \sum_{i=0}^\ell \gamma_i \ .
\end{equation*}
The degrees of the rows of $\Phi(M)$ satisfy
\begin{align*}
  \deg \Phi_\w(\m_0) &= \max_i \{\gamma_i + \deg s_i\} \leq \param \ , \\
  \deg \Phi_\w(\m_i) &= \gamma_i + \deg g_i & \text{for } i \geq 1.
\end{align*}
Thus, $\OD{\Phi_\w(M)} \leq \param - \gamma_0$.
With $\maxdeg(\Phi_\w(M)) \leq \param$, \cref{thm:ms_complexity_general} implies a complexity of $O(\ell^2 \param^2)$, assuming $\ell \in O(\param)$.
Note that the straightforward bound on \cref{alg:ms} yields $O(\ell^3 \param^2)$.
\hfill \IEEEQEDhere
\end{example}

\begin{example}[Mulders--Storjohann for the Interpolation Step in decoding MV codes]\\%
\label{ex:MS_MV_complexity}%
\cref{line:mv_interpol} of \cref{alg:mv_interpol} is a row reduction of $\Phi_\w(M)$, as defined in \eqref{eq:M_MV} on page \pageref{eq:M_MV}, whose degrees of the nonzero entries are component-wise upper bounded by:
\begin{align*}
\begin{pmatrix}
n &  & & & &\\
n & (k-1) & & &\\
n & & 2(k-1) & &\\
\vdots & & & \ddots & \\
n & & & & \ell(k-1)
\end{pmatrix}
\end{align*}
Thus $\OD{\Phi_\w(M)} \leq \deg \Phi_\w(M) - n - \binom {\ell+1} 2(k-1) \leq \ell n$.
Using \cref{thm:ms_complexity_general}, the complexity in operations over $\Fqs$ becomes $O(\ell^3 n^2)$.
\end{example}

\section{Faster Row Reduction on Matrices having Special Forms}
\label{sec:FasterRowRed}

In this section, we will investigate improved row reduction algorithms for matrices of special forms.
The main goals are to improve the running time of row reducing the matrices appearing in the decoding settings of \cref{sec:gab_codes} and \cref{sec:mv_codes}, but the results here apply more broadly.

\subsection{Shift Register Problems: The Demand-Driven Algorithm}\label{sec:demanddriven}
\label{ssec:DD}

Our first focus is to improve the MgLSSR case of \Fref[vario]{alg:shiftregisteralg}, where we are to row reduce $\Phi_\w(M)$, given by \eqref{eqn:phiM}: \cref{alg:dd} is a refinement of \cref{alg:ms} which is asymptotically faster when all $g_i$ are of the form $x^{d_i} + a_i$ for $a_i \in \K$.
Though the refinement is completely analogous to that of \cite{nielsen2013generalised} for the $\K[x]$ case, no complete proof has appeared in unabridged, peer-reviewed form before, so we give full proofs of the $\R$ case here.
We begin with a technical lemma:

\def\p#1{\tilde #1}
\def\prev{\code{previous}}
\def\continue{\code{continue}}
\def\ass{\leftarrow}

\begin{lemma}
  \label{lem:msVar}
  Consider an instance of \cref{prob:skewpade} and \cref{alg:ms} with input $\Phi_\w(M)$ of \eqref{eqn:phiM}.
  Let $\p g_j = g_j x^{\gamma_j}$.
  Consider a variant of \cref{alg:ms} where, after a simple LP-transformation $i$ on $j$, which replaces $\vec v_j$ with $\vec v'_j$, we instead replace it with $\vec v''_j = (v'_{j,0}, v'_{j,1} \modop \p g_1, \ldots, v'_{j,\ell} \modop \p g_\ell)$.
  This does not change the correctness of the algorithm or the upper bound on the number of simple LP-transformations performed.
\end{lemma}
\begin{proof}
  Correctness follows if we can show that each of the $\ell$ modulo reductions could have been achieved by a series of row operations on the current matrix $V$ after the simple LP-transformation producing $\vec v_j'$.
  For each $h \geq 1$, let $\vec g_h = (0,\ldots,0,\p g_h,0,\ldots,0)$, with position $h$ non-zero.

  During the algorithm, we will let $J_h$ be a subset of the current rows in $V$ having two properties: that $\vec g_h$ can be constructed as an $\R$-linear combination of the rows in $J_h$; and that each $\vec v \in J_h$ has $ \vecVal(\vec v) \leq \vecVal(\vec g_h)$.
  Initially, $J_h = \{ \vec g_h \}$.

  After simple LP-transformations on rows not in $J_h$, the $h$'th modulo reduction is therefore allowed, since $\vec g_h$ can be constructed by the rows in $J_h$.
  On the other hand, consider a simple LP-transformation $i$ on $j$ where $\vec v_j \in J_h$, resulting in the row $\vec v_j'$.
  Then the $h$'th modulo reduction has no effect since $\vecVal(\vec v_j') < \vecVal(\vec v_j) \leq \vecVal(\vec g_h)$.
  Afterwards, $J_h$ is updated as $J_h = J_h \setminus \{ \vec v_j \} \cup \{ \vec v_j', \vec v_i \}$.
  We see that $J_h$ then still satisfies the two properties, since $\vecVal(\vec v_i) \leq \vecVal(\vec v_j) \leq \vecVal(\vec g_h)$.

  Since $\vecVal(\vec v''_j) \leq \vecVal(\vec v'_j)$ the proof of \cref{thm:ms_complexity_general} shows that the number of simple LP-transformations performed is still bounded by $(\ell+1)(\OD{V}+\ell+1)$.
\end{proof}

\begin{algorithm}[t]
  \caption{Demand--Driven algorithm for MgLSSR}
  \label{alg:dd}
  \begin{algorithmic}[1]
    \Require{Instance of \cref{prob:skewpade}. $\p s_j \ass s_{1,j}x^{\gamma_j},\ \p g_j \ass g_jx^{\gamma_j}$ for $j=1,\ldots,\ell$. \;}
    \Ensure{The zeroth column of a basis of $\M$ of \eqref{eq:module_vectors} in $\w$-shifted weak Popov form. \;}
    \State $(\eta, h) \ass (\deg, \LPc)$ of $(x^{\gamma_0}, \p s_1, \ldots, \p s_\sigma)$.
    \Ifline{$h = 0$}{\Return $(1, 0, \ldots, 0)$.}
    \State $(\lambda_0,\ldots,\lambda_\ell) \ass (x^{\gamma_0},0,\ldots,0)$.
    \State $\alpha_jx^{\eta_j} \ass $ the leading monomial of $\p g_j$ for $j=1,\ldots,\ell$.
    \While{$\deg \lambda_0 \leq \eta$}
      \State $\alpha \ass $ coefficient to $x^\eta$ in $(\lambda_0\p s_h \mod \p g_h)$.
      \label{line:dd_coefficient}
      \If{$\alpha \neq 0$}
      \label{line:dd_ifstatement}
        \Ifline{$\eta < \eta_h$}{swap $(\lambda_0,\alpha,\eta)$ and $(\lambda_h, \alpha_h, \eta_h)$.}
        \State $\lambda_0 \ass \lambda_0 - \alpha/\theta^{\eta-\eta_h}(\alpha_h) x^{\eta-\eta_h} \lambda_h$.
      \label{line:dd_transformation}
      \EndIf
      \State $(\eta, h) \ass  (\eta, h-1) \textbf{ if } h>1 \textbf{ else } (\eta-1, \ell)$.
    \EndWhile
    \State \Return $\big(\lambda_0x^{-\eta_0},\ldots,\lambda_\ell x^{-\eta_0}\big)$.
  \end{algorithmic}
\end{algorithm}

\begin{theorem}\label{thm:dd_correctness}
  \cref{alg:dd} is correct.
\end{theorem}
\begin{proof}
  We first prove that an intermediary algorithm, \cref{alg:mid-dd}, is correct using the correctness of \cref{alg:ms}, and then prove the correctness of \cref{alg:dd} using \cref{alg:mid-dd} and \cref{lem:msVar}.
  Starting from \cref{alg:ms} with input $\Phi_\w(M)$, then \cref{alg:mid-dd} is obtained by two simple modifications:
  Firstly, note that initially, when $V := \Phi_\w(M)$, then $\LP{\vec v_h} = h$ for $h \geq 1$, and therefore the only possible simple LP-transformation must involve $\vec v_0$.
  We can maintain this property as a loop invariant throughout the algorithm by swapping $\vec v_0$ and $\vec v_{\LP{\vec v_0}}$ when applying a simple LP-transformation $\LP{\vec v_0}$ on $0$.

  The second modification is to maintain $(\eta, h)$ as an upper bound on the $(\deg, \LPc)$ of $\vec v_0$ throughout the algorithm: we initially simply compute these values.
  Whenever we have applied a simple LP-transformation on $\vec v_0$ resulting in $\vec v_0'$, we know by \cref{lem:val_dec} that $\vecVal(\vec v_0') < \vecVal(\vec v_0)$.
  Therefore, either $\deg \vec v_0' < \eta$ or $\deg \vec v_0' = \eta \land \LP{\vec v_0'} < h$.
  This is reflected in a corresponding decrement of $(\eta, h)$.

  As a loop invariant we therefore have $\vecVal(\vec v_0) \leq \eta(\ell+1) + h$.
  After an iteration, if this inequality is sharp, it simply implies that the $\alpha$ computed in the following iteration will be 0, and $(\eta, h)$ will be correspondingly decremented once more.
  Note that we never set $h = 0$: when $\LP{\vec v_0} = 0$ then $V$ must be in weak Popov form (since we already maintain $\LP{\vec v_h} = h$ for $h>0$).
  At this point, the while-loop will be exited since $\deg v_0 > \eta$.

\begin{algorithm}[t]
 \floatname{algorithm}{Intermediate Algorithm}
  \caption{for the correctness proof of \cref{alg:dd}}
  \label{alg:mid-dd}
  \begin{algorithmic}[1]
    \Require{Instance of \cref{prob:skewpade}.  $V \ass \Phi_\w(M)$ with $M$ as in \eqref{eqn:phiM}.}
    \Ensure{A basis $V'$ of $\M$ of \eqref{eq:module_vectors} in $\w$-shifted weak Popov form. \;}
    \State $(\eta, h) \ass (\deg, \LPc)$ of $\vec v_0$.
    \Ifline{$h = 0$}{\Return $\Phi_\w^\mo(V)$.}
    \While{$\deg v_{0,0} \leq \eta$}
      \State $\alpha \ass $ coefficient to $x^\eta$ in $v_{0,h}$.
      \label{line:mid-dd_coefficient}
      \If{$\alpha \neq 0$}
      \label{line:mid-dd_ifstatement}
        \State $\eta_h \ass \deg \vec v_h$.
        \label{line:mid-dd_degv}
        \State $\alpha_h \ass $ coefficient to $x^{\eta_h}$ in $v_{h,h}$.
        \Ifline{$\eta < \eta_h$}{swap $(\vec v_0, \alpha, \eta)$ and $(\vec v_h, \alpha_h, \eta_h)$.}
        \State $\vec v_0 \ass \vec v_0 - \alpha/\theta^{\eta-\eta_h}(\alpha_h) x^{\eta-\eta_h} \vec v_h$.
      \label{line:mid-dd_transformation}
      \EndIf
      \State $(\eta, h) \ass  (\eta, h-1) \textbf{ if } h>1 \textbf{ else } (\eta-1, \ell)$.
    \EndWhile
    \State \Return $\Phi_\w^\mo(V)$.
  \end{algorithmic}
\end{algorithm}

  \cref{alg:mid-dd} is then simply the implementation of these modifications, and writing out in full what the simple LP-transformation does to $\vec v_0$.
  This proves that \cref{alg:mid-dd} is operationally equivalent to \cref{alg:ms} with input $\Phi_\w(M)$.

  For obtaining \cref{alg:dd} from \cref{alg:mid-dd}, the idea is to store only the necessary part of $V$ and compute the rest on demand.
  Firstly, by \cref{lem:msVar} correctness would be maintained if the simple LP-transformation on \cref{line:mid-dd_transformation} of \cref{alg:mid-dd} was followed by the $\ell$ modulo reductions.
  In that case, we would have $v_{0,h} = (v_{0,0}\p s_h \mod \p g_h)$, so storing only $v_{0,0}$ suffice for reconstructing $\vec v_0$.
  Consequently we store the first column of $V$ in \cref{alg:dd} as $(\lambda_0,\ldots,\lambda_\ell)$.
  \cref{line:dd_coefficient} of \cref{alg:dd} is now the computation of the needed coefficient of $v_{0,h}$ at the latest possible time.

  As $\deg \vec v_h$ is used in \cref{line:mid-dd_degv} of \cref{alg:mid-dd}, we need to store and maintain this between iterations; this is the variables $\eta_1,\ldots,\eta_\ell$.
  To save some redundant computation of coefficients, the $x^{\eta_h}$-coefficient of $v_{h,h}$ is also stored as $\alpha_h$. 

  This proves that \cref{alg:dd} is operationally equivalent to \cref{alg:mid-dd}, which finishes the proof of correctness.
\end{proof}

\begin{proposition}
  \label{prop:dd-complexity}
  \cref{alg:dd} has computational complexity $O(\ell \param^2 + \sum_{h=1}^\ell \sum_{\eta=0}^{\param-1} T_{h,\eta})$, where
  $\param = \max_i \{\gamma_i +\deg g_i\}$ and $T_{h,\eta}$ bounds the complexity of running \cref{line:dd_coefficient} for those values of $h$ and $\eta$.
\end{proposition}
\begin{proof}
  Clearly, all steps of the algorithm are essentially free except \cref{line:dd_coefficient} and \cref{line:dd_transformation}.
  Observe that every iteration of the while-loop decrease an \emph{upper bound} on the value of row 0, whether we enter the if-branch in \cref{line:dd_ifstatement} or not.
  So by the arguments of the proof of \cref{thm:ms_complexity_general}, the loop will iterate at most $O(\ell \param)$ times in which each possible value of $(h, \eta) \in \{1,\ldots,\ell\} \times \{0,\ldots,\param-1\}$ will be taken at most once.
  Each execution of \cref{line:dd_transformation} costs $O(\param)$ since the $\lambda_j$ all have degree at most $\param$.
\end{proof}

It is possible to use \cref{prop:dd-complexity} to show that \cref{alg:dd} is efficient if e.g.~all the $g_i$ have few non-zero monomials\footnote{%
  In the conference version of this paper \cite{li2015solving}, we erroneously claimed a too strong statement concerning this. However, one \emph{can} relate the complexity of \cref{alg:dd} to the number of non-zero monomials of $g_i$, as long as all but the leading monomial have low degree; however the precise statement becomes cumbersome and is not very relevant for this paper.
  }%
. We will restrict ourselves to a simpler case which nonetheless has high relevance for coding theory:
\begin{theorem}
  \label{thm:dd_fastonsimple}
  \cref{alg:dd} can be realised with complexity $O(\ell \mu^2)$ if $g_i = x^{d_i} + a_i$ for $a_i \in \Fq$ for all $i$, where $\param = \max_i \{\gamma_i +\deg g_i\}$.
\end{theorem}
\begin{proof}
  We will bound $\sum_{\eta=0}^{\param-1} T_{h,\eta}$ of \cref{prop:dd-complexity}.
  Note first that for any $\eta$, the coefficient $\alpha$ to $x^\eta$ in $(\lambda \tilde{s}_h \modop \tilde{g}_h)$ equals the coefficient to $x^{\eta-\gamma_h}$ of $(\lambda s_h \modop g_h)$, so considering $\gamma_h = 0$ suffice.
  Now if $\eta \geq d_h$ then $\alpha = 0$ and can be returned immediately.
  If $\eta < d_h$, then due to the assumed shape of $g_i$, $\alpha$ is a linear combination of the coefficients to $x^{\eta}, x^{\eta+d_h}, \ldots, x^{\eta + t d_h}$ in $\lambda_0 s_h$, where $t = \left\lfloor \frac {\param - \eta} {d_h} \right\rfloor$.
  Each such coefficient can be computed by convolution of $\lambda_0$ and $s_h$ in $O(\param)$, so it costs $O(\frac{\param^2} d)$ to compute $\alpha$.
  Summing over all choices of $\eta$, we have $\sum_{\eta=0}^{\param-1} T_{h,\eta} \in O(\param^2)$ and the theorem follows from \cref{prop:dd-complexity}.
\end{proof}

\subsection{Weak Popov Walking}
\label{ssec:wpfwalk}

The goal of this section is to arrive at a faster row reduction algorithm for the matrices used for decoding Mahdavifar--Vardy codes in \cref{sec:mv_codes}.
However, the algorithm we describe could be of much broader interest: it is essentially an improved way of computing a $\vec w$-weak Popov form of a matrix which is already in $\vec w'$-weak Popov form, for a shift $\vec w'$ which is not too far from $\vec w$.
Inspired by ``Gr\"obner walks'', we have dubbed this strategy ``weak Popov walking''.
Each ``step'' of the walk can be seen as just \cref{alg:ms} but where we carefully choose which LP-transformations to apply each iteration, in case there is choice.

This strategy would work completely equivalently for the $\K[x]$ case.
However, to the best of our knowledge, that has not been done before.

In this section we will extensively discuss vectors under different shifts.
To ease the notation we therefore introduce shifted versions of the following operators: $\LP[\w]{\v} := \LP{\Phi_\w(\v)}$ as well as $\deg_\w(\v) := \deg \Phi_\w(\v)$.

We begin by \cref{alg:wpfwalk} that efficiently ``walks'' from a weak Popov form according to the shift $\w$ into one with the shift $\w + (1,0,\ldots,0)$.
The approach can readily be generalised to support increment on any index, but we do not need it for the decoding problem so we omit the generalisation to simplify notation.

\def\hv{{\vec{\hat v}}}
\def\hw{{\vec{\hat w}}}

\begin{algorithm}[H]
  \caption{Weak Popov Walking}
  \label{alg:wpfwalk}
  \begin{algorithmic}[1]
    \Require{Shift $\w \in \NN^m$ and matrix $V \in \R^{m \times m}$ in $\w$-shifted weak Popov form.}
    \Ensure{Matrix in $\hw$-shifted weak Popov form spanning the same $\R$-row space as $V$, where $\hw = \w + (1,0,\ldots,0)$.}
    \State $h_i \ass \LP[\w]{\v_i}$, for $i=0,\ldots,m-1$.
    \State $I \ass $ indexes $i$ such that $\LP[\hw]{\v_i} = 0$.
    \State $[i_1,\ldots,i_s] \ass I$ sorted such that $h_{i_1} < h_{i_2} < \ldots < h_{i_s}$.
    \State $t \ass i_1$.
    \For{$i = i_2, \ldots, i_s$}
      \If{$\deg v_{t,0} \leq \deg v_{i,0}$}
        \State Apply a simple transformation $t$ on $i$ at position $0$ in $V$.
          \label{line:wpfwalk:trans1}
      \Else
        \State Apply a simple transformation $i$ on $t$ at position $0$ in $V$.
          \label{line:wpfwalk:trans2}
        \State $t \ass i$.
      \EndIf
    \EndFor
    \State \Return $V$.
  \end{algorithmic}
\end{algorithm}

\begin{theorem}
  \label{thm:wpfwalk_correct}
  \cref{alg:wpfwalk} is correct.
\end{theorem}
\begin{proof}
  Denote in this proof $V$ as the input and $\hat V$ as the output of the algorithm.
  The algorithm performs a single sweep of simple transformations, modifying only rows indexed by $I$: in particular, if $\v_i,\hv_i$ are the rows of $V$ respectively $\hat V$, then either $\hv_i = \v_i$, or $\hv_i$ is the result of a simple transformation on $\v_i$ by another row $\v_j$ of $V$ and $i,j \in I$.
  All the $h_i$ are different since $V$ is in $\w$-shifted weak Popov form.
  We will show that the $\hw$-shifted leading positions of $\hat V$ is a permutation of the $h_i$, implying that $\hat V$ is in $\hw$-shifted weak Popov form.

  Note first that for any vector $\v \in \R^m$ with $\LP[\w]{\v} \neq \LP[\hw]{\v}$, then $\LP[\hw]{\v} = 0$, since only the degree of the $0$'th position of $\Phi_\hw(\v)$ is different from the corresponding position of $\Phi_\w(\v)$.
  For each $i \in \{0,\ldots,m-1\} \setminus I$ we have $\hv_i = \v_i$ and so $\LP[\hw]{\hv_i} = h_i$.
  And of course for each $i \in I$ we have $\LP[\hw]{\v_i} = 0$.
  This implies for each $j \in I$ that:
  \begin{equation}
    \label{eqn:wpfwalk:degreerel}
    \deg v_{j,0} + w_0 = \deg v_{j,h_j} + w_{h_j} \ .
  \end{equation}

  Consider first an index $i \in I$ for which \cref{line:wpfwalk:trans1} was run, and let $t$ be as at that point. 
  This means $\hv_i = \v_i + \alpha x^\delta \v_t$ for some $\alpha \in \K$ and $\delta = \deg v_{i,0} - \deg v_{t,0}$.
  Note that the if-condition ensures $\delta \geq 0$ and the simple transformation makes sense.
  We will establish that $\LP[\hw]{\hv_i} = h_i$.
  Since we are performing an LP-transformation, we know that $\deg_\hw \hv_i \leq \deg_\hw \v_i$, so we are done if we can show that $\deg \hat v_{i,h_i} = \deg v_{i,h_i}$ and $\deg \hat v_{i, k} + w_k < \deg v_{i,h_i} + w_{h_i}$ for $k > h_i$.
  This in turn will follow if $\alpha x^\delta \v_t$ has $\hw$-weighted degree less than $\deg v_{i,h_i} + w_{h_i}$ on all position $k \geq h_i$.
  
  Due to $\LP[\w]{\v_t} = h_t$ and \eqref{eqn:wpfwalk:degreerel} for index $t$ then for any $k > h_t$:
  \begin{equation}
    \label{eqn:wpfalk:degreerel2}
    \deg v_{t, k} + w_k < \deg v_{t, h_t} + w_{h_t} = \deg v_{t,0} + w_0 \ .
  \end{equation}
  Using $\deg v_{t,0} + \delta = \deg v_{i,0}$ and \eqref{eqn:wpfwalk:degreerel} for index $i$, we conclude that 
  \[
    \deg v_{t, k} + w_k + \delta < \deg v_{i,0} + w_{0} = \deg v_{i,h_i} + w_{h_i} \ .
  \]
  Since $h_t < h_i$ by the ordering of the $i_\star$, this shows that $\deg v_{i, k} + w_k + \delta < \deg v_{i,h_i} + w_{h_i}$ for $k \geq h_i$.
  These are the degree bounds we sought and so $\LP[\hw]{\hv_i} = h_i$.

  Consider now an $i \in I$ for which \cref{line:wpfwalk:trans2} was run, and let again $t$ be as at that point, before the reassignment.
  The situation is completely reversed according to before, so by analogous arguments $\LP[\hw]{\hv_t} = h_i$.

  For the value of $t$ at the end of the algorithm, then clearly $\LP[\hw]{\hv_t} = 0$ since the row was not modified.
  Since we necessarily have $h_{i_1} = 0$, then $\LP[\hw]{\hv_t} = h_{i_1}$.
  Thus every $h_i$ becomes the $\hw$-leading position of one of the $\vec v_j$ exactly once.
  But the $h_i$ were all different, and so $\hat V$ is in $\hw$-shifted weak Popov form.
\end{proof}

\begin{proposition}
  \label{prop:wfpwalk_complexity}
  \cref{alg:wpfwalk} performs at most
  \[
    \textstyle
    O\big(m\deg \det(V) + \sum_{i < j}|w_i - w_j| + m^2\big)
  \]
 operations over $\R$.
\end{proposition}
\begin{proof}
  We will bound the number of non-zero monomials which are involved in simple transformations.
  As remarked in the proof of \cref{thm:wpfwalk_correct}, all simple transformations are done using distinct rows of the input matrix, so it suffices to bound the total number of monomials in the input matrix $V$.

  Since we are then simply counting monomials in $V$, we can assume w.l.o.g.~that $w_0 \leq w_1 \leq \ldots \leq w_{m-1}$, and since the input matrix $V$ was in $\w$-shifted weak Popov form, assume also w.l.o.g~that we have sorted the rows such that $\LP[\w]{\v_i} = i$.
  Since $\OD{\Phi_\w(V)} = 0$ we have
  \[
    \deg \det \Phi_\w(V) = \deg_\w V 
    \qquad \textrm{ that is }\qquad 
    \textstyle \deg_\w V = \deg \det V + \sum_i w_i \ .
  \]
  We can therefore consider the assignment of $\deg_\w$ to the individual rows of $V$ under these constraints that will maximise the possible number of monomials in $V$.
  We cannot have $\deg_\w \v_i < w_i$ since $\LP[\w]{\v_i} = i$.
  It is easy to see that the worst-case assignment is then to have exactly $\deg_\w \v_i = w_i$ for $i=0,\ldots,m-2$ and $\deg_\w \v_{m-1} = \deg \det V + w_{m-1}$.
  In this case, for $i < m-1$ then $\deg v_{i,j} \leq  w_i - w_j$ if $j \leq i$ and $v_{i,j} = 0$ if $j > i$, so the number of monomials can then be bounded as
  \begin{IEEEeqnarray*}{rCl}
    && \left( \sum_{i=0}^{m-2} \sum_{j=0}^{i}(w_i - w_j + 1) \right) + \left(\sum_{j=0}^{m-1}(\deg\det V + w_{m-1} - w_j + 1) \right) \\
   &\leq& m^2 + \sum_{i < j}(w_j - w_i) + m\deg\det V
   \ .
   \\[-3em]
  \end{IEEEeqnarray*}
\end{proof}

The idea is now to iterate \cref{alg:wpfwalk} to ``walk'' from a matrix that is in weak Popov form for one shift $\w$ into another one $\hw$.
Row reducing the matrix for the MV codes can be done as \cref{alg:mv_wpfwalk}.

\begin{algorithm}[H]
  \caption{Find MV Interpolation Polynomial by Weak Popov Walk}
  \label{alg:mv_wpfwalk}
  \begin{algorithmic}[1]
    \Require{Instance of \cref{prob:mv_interpolation} and the matrix $V \ass M$ of \eqref{eq:M_MV} on page \pageref{eq:M_MV}}
    \Ensure{A $\vec w$-shifted weak Popov form of $M$}
    \State $\vec w = \big(0, (k-1), 2(k-1),\ldots, \ell(k-1) \big)$.
    \State $\vec w' = \vec w + \big(0, n, n, \ldots, n\big)$.
      \label{line:mv_wpfwalk:shiftsinit}
    \For{$i = 0, ..., n-1$}
      \State $V \ass \mathrm{WeakPopovWalk}(V, \vec w')$.
      \State $\vec w' \ass \vec w' + (1,0,\ldots,0)$.
    \EndFor
    \State \Return $V$
  \end{algorithmic}
\end{algorithm}

\begin{theorem}\label{thm:wpf_correctness}
  \cref{alg:mv_wpfwalk} is correct.
  It has complexity $O(\ell n^2)$ over $\Fqs$.
\end{theorem}
\begin{proof}
  Note that $M$ is in $\w'$-shifted weak Popov form, where $\w'$ is as on \cref{line:mv_wpfwalk:shiftsinit}.
  Thus by the correctness of \cref{alg:wpfwalk}, then $V$ at the end of the algorithm must be in $\big(\w + (n, \ldots, n)\big)$-shifted weak Popov form.
  Then it is clearly also in $\w$-shifted weak Popov form.
  For the complexity, the algorithm simply performs $n$ calls to \cref{alg:wpfwalk}.
  We should estimate the quantity $\sum_{i < j}|w_i - w_j|$, which is greatest in the first iteration.
  Since \cref{prob:mv_interpolation} posits $n > \binom {\ell+1} 2 (k-1)$, we can bound the sum as:
  \[
    \sum_{j=1}^{\ell}(n + j(k-1)) + \sum_{1 \leq i < j}(j-i)(k-1) < \ell n + (\ell+1)\tbinom {\ell+1} 2 (k-1) \in O(\ell n) \ .
  \]
  Since $\deg \det(V) = \deg\det(M) = n$ then by \cref{prop:wfpwalk_complexity} each of the calls to \cref{alg:wpfwalk} therefore costs at most $O(\ell n)$.
\end{proof}

\section{Conclusion}\label{sec:conclusion}

We have explored row reduction of skew polynomial matrices.
For ordinary polynomial rings, row reduction has proven a useful strategy for obtaining flexible, efficient while conceptually simple decoding algorithms for Reed--Solomon and other code families.
Our results introduce the methodology and tools aimed at bringing similar benefits to Gabidulin, Interleaved Gabidulin, Mahdavifar--Vardy, and other skew polynomial-based codes.
We used those tools in two settings.
We solved a general form of multiple skew-shift register synthesis (cf.~\cref{prob:skewpade}), and applied this for decoding of Interleaving Gabidulin codes in complexity $O(\ell \mu^2)$, see \cref{thm:shiftregister_complexity}.
For Mahdavifar--Vardy codes (cf.~\cref{prob:mv_interpolation}), we gave an interpolation algorithm with complexity $O(\ell n^2)$, see~\cref{thm:mv_complexity}.

We extended and analysed the simple and generally applicable Mulders--Storjohann algorithm to the skew polynomial setting.
In both the studied settings, the complexity of that algorithm was initially not satisfactory, but it served as a crucial step in developing more efficient algorithms.
For multiple skew-shift register synthesis, we were able to obtain a good complexity for a more general problem than previously.
For the Mahdavifar--Vardy codes, the improved algorithm was in the shape of a versatile ``Weak Popov Walk'', which could potentially apply to many other problems.
In all previously studied cases, we matched the best known complexities \cite{sidorenko2011skew, xie_linearized_2013} that do not make use of fast multiplication of skew polynomials.

Based on a preprint of this paper, in \cite{puchinger2016alekhnovich} it is shown how to further reduce the complexity for decoding Interleaved Gabidulin codes using a divide-\&-conquer version of Algorithm~\ref{alg:ms}, matching the complexity of \cite{sidorenko2014fast}.

The weak Popov form has many properties that can be beneficial in a coding setting, and which we did not yet explore.
For instance, it allows to easily enumerate all ``small'' elements of the row space: that could e.g.~be used to enumerate \emph{all} solutions to a shift register problem, allowing a chase-like decoding of Interleaved Gabidulin codes beyond half the minimum distance.

\begin{acknowledgement}
The authors would like to thank the anonymous reviewers for suggestions that have substantially improved the clarity of the paper.
\end{acknowledgement}

\bibliographystyle{spmpsci}
\bibliography{ModMinIG}

\end{document}